\documentclass[preprint,12pt]{elsarticle}
\usepackage[utf8]{inputenc}
\usepackage[T1]{fontenc}
\usepackage{amsfonts}
\usepackage{amsmath,amsthm}
\usepackage{amssymb}
\usepackage{color,soul}
\usepackage{textcomp}
\usepackage{geometry}
\usepackage{enumitem}
\usepackage{todonotes}
\usepackage{graphicx}
\usepackage{subcaption}
\usepackage{epstopdf}
\usepackage{float}
\usepackage{tikz}
\usepackage{todonotes}
\usepackage{natbib}
\usepackage{graphicx}

\makeatletter
\newcommand{\vast}{\bBigg@{4}}
\newcommand{\Vast}{\bBigg@{5}}
\makeatother

\newtheorem{theorem}{Theorem}
\newtheorem{proposition}{Proposition}

\newtheorem{remark}{Remark}

\title{Exact solutions for a Solow-Swan model with non-constant returns to scale}
\author{Nicol\`o Cangiotti{$^{1}$} \& Mattia Sensi{$^{2}$}\\[1em]
\small $^1$\emph{University of Pavia, Department of Mathematics, via Ferrata 5, \\ \small 27100 Pavia (PV), Italy. Email:} \texttt{nicolo.cangiotti@unipv.it}\\
 \small $^2$\emph{University of Trento, Department of Mathematics, via Sommarive 14,\\ \small 38123 Trento (TN), Italy. Email:} \texttt{mattia.sensi@unitn.it}}

\date{}

\begin{document}

\begin{abstract}
The Solow-Swan model is shortly reviewed from a mathematical point of view. By considering non-constant returns to scale, we obtain a general solution strategy. We then compute the exact solution for the Cobb-Douglas production function, for both the classical model and the von Bertalanffy model. Numerical simulations are provided.
\vspace{5mm}
\small

\smallskip

\end{abstract}

\maketitle
\section{Introduction}

The Solow-Swan model plays an important role in neoclassical economics. Even though more than 60 years have passed since it was developed, independently, by Robert Solow \cite{Solow56} and  Trevor Swan \cite{Swan56} in 1956, the model is still being analyzed and generalized, as evidenced by a large literature, which involves many fields of studies  \cite{Rubio2000,dohtani2010growth,Farmer2001,gandolfo1997, guerrini2006solow,Halsmayer2014,Kulikov2019,Lund2018}. 

This work is devoted to the deepening of the mathematical point of view of the model. In particular, we are interested in investigating a model with weaker conditions on the returns to scale than the usual ones (see, e.g., \cite{guerrini2006solow}). In fact, we are going to relax the hypothesis of constant returns to scale, which in the classical model allows to rewrite the production function as a function of the output per effective unit of labour; instead, we let the production function to have increasing or decreasing returns to scale. We obtain a non-autonomous first order differential equation, for which we provide the exact solution.

The paper is organized as follows. In Section 2, we present the Solow-Swan model, with a focus on the Cobb-Douglas production function; moreover, we study the non-constant returns to scale case, obtaining the exact solution for the model. In Section 3, we explore a different model, namely the von Bertalanffy model, by using the same techniques. Section 4 is devoted to numerical analysis, to better understand the behaviour of such solutions.  Finally, in Section 5 we shall suggest some perspectives for future research.
\smallskip

\section{The classical model}
\label{S2}
As highlighted in the introduction, Solow-Swan models have a key role in neoclassical growth theory. Let us denote by $C^2(\mathbb{R}^2)$ the class of twice continuously differentiable functions $F:\mathbb{R}^2\to\mathbb{R}$. We actually can restrict the study to $\mathbb{R}^2_+$, i.e. the first quadrant, the only economically relevant subset in this setting.
The mathematics of the model is based on the hypothesis that a production function $F(x_1,x_2)$
satisfies the following conditions:

\begin{equation}
    \label{inada}
   \frac{\partial F}{\partial x_i}> 0, \quad 
   \frac{\partial ^2F}{\partial x_i\partial x_j}< 0,  \quad 
   \lim_{x_i \to 0^+}  \frac{\partial F}{\partial x_i}= +\infty, \quad 
   \lim_{x_i \to \infty} \frac{\partial F}{\partial x_i}= 0,
\end{equation}
 for $i,j=1,2$.
\smallskip
In literature, conditions \eqref{inada}, which are aimed at ensuring the existence of an unique stable steady state in a neoclassical growth model, are called \emph{Inada condtions}. For further details and properties about the Inada conditions, we refer to \cite{barelli2003,inada63,litina2008,Taka1985,uzawa63}.

Classically, the variables $x_1$ and $x_2$ are denoted with $K$ and $L$, respectively. We switch to this notation for the remainder of the article.

Moreover, the quite stringent assumption that $F$ has constant returns to scale is rather frequent in many studies. In fact, thanks to such a hypothesis, it is not hard to obtain an exact solution for the ODE describing the model, at least in its most famous \emph{autonomous} form.\\
Since it is a useful step towards our more general construction, we briefly recall this strategy. Let us suppose that the rate of change of $K$ is proportional to $F$, and the labor force grows exponentially; such a setting can be described by the following system:
\begin{align*}
    \frac{\textnormal{d}K}{\textnormal{d}t}&=sF(K,L),\\
    \frac{\textnormal{d}L}{\textnormal{d}t}&=\gamma L,
\end{align*}
with $s,\lambda >0$ constants. Thus, since the equation for $L$ is autonomous and easily solved, we focus our attention on the ODE describing the evolution in time of $K$, which is, explicitly:
\begin{equation}
\label{FirstEq}
\frac{\textnormal{d}K}{\textnormal{d}t}=sF(K,L).
\end{equation}
We notice that the constant return to scale hypothesis implies that 
\begin{equation}
\label{ConstRTS}
    F(\lambda K, \lambda L)=\lambda F(K,L).
\end{equation}
Dividing both sides of \eqref{FirstEq} by $L$, the equation becomes
\begin{equation}
\label{SecondEq}
    \frac{1}{L}\frac{\textnormal{d}K}{\textnormal{d}t}=sF \left ( \frac{K}{L},1 \right ).
\end{equation}
Let us now consider the following derivative:
\begin{equation}
    \label{derivativeKL}
    \frac{\textnormal{d}}{\textnormal{d}t}\left (\frac{K}{L} \right)=\frac{1}{L}\frac{\textnormal{d}K}{\textnormal{d}t}-K\frac{\textnormal{d}L}{\textnormal{d}t}\frac{1}{L^2}=\frac{1}{L}\frac{\textnormal{d}K}{\textnormal{d}t}-\gamma\frac{K}{L}.
\end{equation}
Combining \eqref{SecondEq} and \eqref{derivativeKL}, and introducing the variable $k:=\frac{K}{L}$, i.e. the capital-labor ratio, and the notation $f(k):=F(k,1)$, we are finally ready to write the classic Solow-Swan model:
\begin{equation}
    \label{SolowSwan1}
    \frac{\textnormal{d}k}{\textnormal{d}t}=sf(k)-\gamma k.
\end{equation}
\smallskip

However, in this paper we shall present a different approach to the Solow-Swan model compared to the one given in \cite{guerrini2006solow}, which is
\begin{equation}
\dot{k}=sf(k)-(\delta+\gamma(t))k,
\label{eqn:guerrinisw}
\end{equation}
where $k$ is the capital-labor ratio, $s$ is the fraction of output which is saved, $\delta$ is the depreciation rate, $f$ is a production function and $\gamma(t)$ is the ratio $\dot{L}/L$; $\dot{k}$ indicates the derivative of $k$ with respect to the time variable $t$, i.e. $\frac{\textnormal{d}}{\textnormal{d}t}k(t)$.
In fact, in \cite{guerrini2006solow}, the author assumed $f$ to have constant return to scale (as in the original model), and $\gamma$ to be variable in time. Conversely, we assume $\gamma$ to be constant, from which we obtain 
\begin{equation}
\frac{\dot{L}}{L}=\gamma \implies L(t)=L_0 e^{\gamma t}.
\label{eqn:labor}
\end{equation}
However, we do \textit{not} assume our production function $f$ to have constant return to scale. Instead, we choose a generic homogeneous production function, namely



\smallskip

\begin{equation}
F(\lambda K, \lambda L)=\lambda^{n} F(K,L).
\label{eqn:cobbyboy}
\end{equation}
This means that, if $n=1$, the function has constant return to scale, if $n<1$ ($n>1$) the function has decreasing (increasing) returns to scale. In particular, we notice that
\begin{equation}
F(K/L,1)=F(L^{-1}K,L^{-1}L)=L^{-n}F(K,L).
\label{eqn:trick}
\end{equation}
Starting from the usual equations
\begin{subequations}
\begin{align}
\dot{K}=&sF(K,L),\label{eqn:start1}\\
\dot{L}=&\gamma L
\label{eqn:start2},
\end{align}
\label{eqn:start}%
\end{subequations}
we can derive a non-autonomous equation for the capital-labor ($K/L$) ratio $k$, as stated in the following proposition.
\begin{proposition}
\label{prop:ratio1}
The ratio $k$ evolves in time obeying the ODE
\begin{equation}
    \dot{k}=sL^{n-1}(t)f(k)-\gamma k,
    \label{eqn:ODE}
\end{equation}
where $f(k):=F(k,1)$ and $L(t)=L_0 e^{\gamma t}$; recall \eqref{eqn:labor} and \eqref{eqn:start2}.
\end{proposition}
\begin{proof}
By direct computation, we notice that
$$
    \frac{\textnormal{d}}{\textnormal{d}t}\frac{K}{L}=\frac{1}{L}\frac{\textnormal{d}K}{\textnormal{d}t}-\frac{K}{L^2}\frac{\textnormal{d}L}{\textnormal{d}t}=\frac{1}{L}\frac{\textnormal{d}K}{\textnormal{d}t}-\gamma \frac{K}{L}.
$$
Now, combining (\ref{eqn:start1}) and (\ref{eqn:trick}), we notice that
$$
\frac{1}{L}\frac{\textnormal{d}K}{\textnormal{d}t}=sL^{n-1}F(K/L,1).
$$
Recalling the definitions of $k=K/L$ and $f(k):=F(k,1)$, we conclude the proof.
\end{proof}

\begin{remark}
There are many standard properties of the following Cauchy problem:
\begin{equation}
\label{eqn:cauchy1}
    \begin{cases}
       \dot{k}=sL^{n-1}(t)f (k)-\gamma k,\\
     k(0)=k_0,
    \end{cases}
\end{equation}
that one can easily obtain by simple observations or by using to use the so-called \emph{Comparison theorems} (for results in that direction see \cite[Sec. 3]{guerrini2006solow} and \cite{ODEBirkRota}).
\end{remark}
\smallskip

\begin{remark}
The results obtained so far are valid for a wide class of production functions; however, in order to proceed with the analysis of the Cauchy problem \eqref{eqn:cauchy1} one needs to specify a production function.
\end{remark}

Our investigation now proceeds with a very natural choice for the production function $f(k)$, i.e. the \emph{Cobb-Douglas production function} \cite{CobbDouglas28}. For the standard Cobb-Douglas production function (in which we fixed, without loss of generality, the total-factor productivity coefficient equal to $1$)
\[
F(K,L)=K^\alpha L^\beta, \quad 0<\alpha \leq 1, \quad 0<\beta \le 1, \quad \alpha+\beta=n,
\]
it is easy to compute the law of the capital-labor ratio $k(t)$:
\begin{equation}
\label{eqn:ODECD}
\dot{k}=sL_0^{n-1} e^{(n-1)\gamma t}k^\alpha-\gamma k.
\end{equation}
The following theorem provides the exact solution for \eqref{eqn:ODECD}.
\begin{theorem}
\label{Thm:1}
Let $k(t)$ be a solution of \eqref{eqn:ODECD}. Then if $n\neq 1$ and $\alpha\neq 1$
\begin{equation}
    k(t)=\left( e^{(\alpha-1)\gamma t} \bigg[s(1-\alpha)L_0^{n-1}\frac{(e^{\gamma \beta t}-1)}{\gamma \beta} +k_0^{1-\alpha} \bigg]\right)^{\frac{1}{1-\alpha}},
    \label{eqn:thmCD}
\end{equation}
where we denote $k_0:=k(0)$.
\end{theorem}
\begin{proof}
Consider \eqref{eqn:ODECD}, which is clearly a Bernoulli differential equation \cite{ODEHairer}. We divide both sides by $k^\alpha$, and apply the substitution $v=k^{1-\alpha}$. Then, after some algebraic steps, \eqref{eqn:ODECD} becomes
$$
\frac{1}{1-\alpha}\dot{v}+\gamma v = sL_0^{n-1} e^{(n-1)\gamma t},
$$
We multiply both sides by $(1-\alpha)$, which brings the equation to a standard form
$$
\dot{v}+(1-\alpha)\gamma v = sL_0^{n-1}(1-\alpha) e^{(n-1)\gamma t}.
$$
Recalling $\beta=n-\alpha$, we apply the well-known formula to solve this first order linear ODE, obtaining \eqref{eqn:thmCD}.

\end{proof}

\begin{remark}
If $n=\alpha+\beta=1$, i.e. if the Cobb-Douglas function has constant return to scale, we recover Thm. 10 of \cite{guerrini2006solow}.
\end{remark}

\begin{remark}
\label{alpha1}
It is clear that for $\alpha=1$ Eq. \eqref{eqn:ODECD} is a  linear differential equation. The solution of such equation is computable by standard methods and, for $k(0)=k_0$, we have
\[
k(t)=k_0 \exp \left(\frac{s L_0^{n-1} \left(e^{(n-1) t}-1\right)}{n-1}-\gamma  t\right).
\]
\end{remark}

\section{The von Bertalanffy model}
In this section, we propose a different model, in which the labor force follows a von Bertalanffy law \cite{vonB}:
\begin{equation}
\label{vonBerta}
    \begin{cases}
    \dot{L}=r\left (L_{\infty}-L \right),\\
    L(0)=0,
    \end{cases}
    \end{equation}
where 
\begin{equation}
L_{\infty}= \lim_{t\to\infty}L(t),
\end{equation}
is a theoretical maximum asymptote size of the labor force, and $r>0$ determines the speed at which the labor force approaches the asymptote. The model was exhaustively studied by Brida and Limas in \cite{Brida2005}, where the authors present many important results for the constant returns to scale case. As in Sect. \ref{S2}, we are going to relax this hypothesis, considering also increasing (and decreasing) returns to scale, and we present the exact solution for the model.
\begin{remark}
The von Bertalanffy equation was widely studied by many authors from different fields. See, for instance, \cite{Colern1978,JM92,Mingari}.
\end{remark}
The first step is to compute the law of the ratio $k$.
\begin{proposition}
The ratio $k$ evolves in time obeying the ODE
\begin{equation}
    \dot{k}=sL^{n-1}(t)f(k)-r k (L_\infty - L(t)),
    \label{eqn:ODE2}
\end{equation}
where $f(k):=F(k,1)$ and, from \eqref{vonBerta}, $L(t)=L_{\infty}-(L_{\infty}-L_0) e^{-r t}$.
\end{proposition}
\begin{proof}
The proof is analogous to the proof of Prop. \ref{prop:ratio1}.
\end{proof}
We consider the Cobb-Douglas production function to proceed with our investigation, thus obtaining the following Cauchy problem:
\begin{equation}
\label{eqn:cauchy2}
\begin{cases}
\dot{k}=s(L_{\infty}-(L_{\infty}-L_0) e^{-r t})^{n-1} k^\alpha-r  (L_{\infty}-L_0) e^{-r t}k,
\\
k(0)=k_0.
\end{cases}
\end{equation}
The solution of the Cauchy problem \eqref{eqn:cauchy2} is given by the following theorem.
\begin{theorem}
Let $k(t)$ be a solution of \eqref{eqn:cauchy2}. Then if $n\neq 1$ and $\alpha \neq 1$
\begin{equation}
k(t)=\Bigg ( e^{-(\alpha -1) (L_{\infty}-L_0) e^{-r t}}  \Bigg ( k_{0}^{1-\alpha} e^{(\alpha -1) (L_{\infty}-L_{0})}-(\alpha -1) \int_0^t \mathcal{L}(\tau) \, d\tau \Bigg )\Bigg)^{\frac{1}{1-\alpha}},
    \label{eqn:thmCD2}
\end{equation}
where
\[
\mathcal{L}(\tau):=\frac{s\left ( L_{\infty}-\left( L_{\infty}-L_0\right )e^{-r\tau} \right )^n\cdot \exp \left [ r\tau+(\alpha-1)(L_{\infty}-L_0)e^{-r\tau}\right ]}{L_{\infty}\left ( e^{r\tau}-1\right )+L_0}.
\]
\end{theorem}
\begin{proof}
The proof is analogous to the proof of Thm. \ref{Thm:1}. In fact, by the same substitution $v=k^{1-\alpha}$, we get the following linear differential equation:
\[
\dot{v}=(1-\alpha ) s \left (L_\infty-(L_\infty-L_0) e^{-r t}\right)^{n-1}+(1-\alpha ) r (L_\infty-L_0)  e^{-r t} v.
\]
Thus, applying the classical formula, we compute the solution.  
\end{proof}
\begin{remark}
A comment analogous to Rmk. \ref{alpha1} can be made. For $\alpha=1$, we have the following solution (which involves hypergeometric functions $_2F_1$), where we introduce, for ease of notation, $L_*:=L_\infty-L_0$:
\begin{gather*}
k(t)=k_0 \cdot  \exp \vast[\left(e^{-r t}-1\right)L_*+ 
\frac{s L_0^{n-1} \left(-\frac{L_0}{L_*}\right)^{1-n} \, _2F_1\left(1-n,1-n;2-n;\frac{L_{\infty}}{L_*}\right)}{(n-1) r}-\\
\frac{s \left(\frac{L_{\infty} e^{r t}}{-L_*}+1\right)^{1-n} \left(-L_* e^{-r t}+L_{\infty}\right)^{n-1} \, _2F_1\left(1-n,1-n;2-n;\frac{e^{r t} L_{\infty}}{L_*}\right)}{(n-1) r}\vast ].
\end{gather*}
For further details on the use of the hypergeometric function in this context see, for instance, \cite{Mingari}.
\end{remark}

\section{Numerical simulations}


In this section we propose some numerical simulations for both the classical and the von Bertalanffy model, for the specific choice of Cobb-Douglas for our production function $f(k)$. 
The results of the classical case agree with the expectations, consistently with neoclassical growth theory with a convergence toward the initial conditions for the decreasing returns to scale and an exponential growth for increasing returns to scale (Figure \ref{Figure_3}). A very interesting output comes from the study of the von Bertalanffy model, as dysplayed in Figure \ref{Figure_4}. The latter seems to level out the differences between the two cases, namely increasing returns to scale and decreasing returns to scale. The following graphs show the behaviour of the capital-labor ratio.

\begin{figure}[ht!]
    \centering
    \begin{minipage}[b]{0.475\textwidth}
            \centering
    \includegraphics[width=0.95\textwidth]{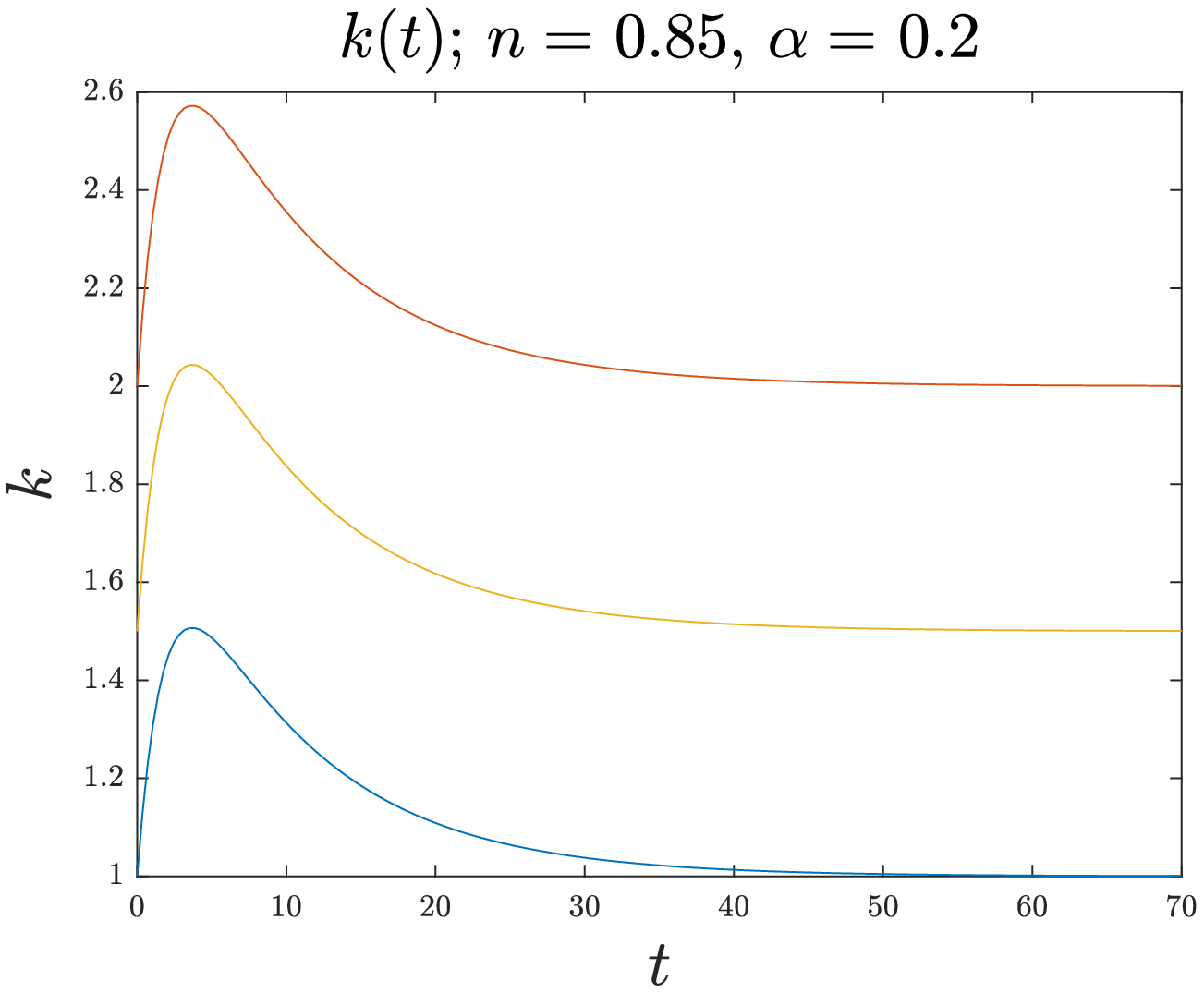}
  \caption*{(a) For $n<1$, we observe a rapid growth, followed by a convergence towards the initial conditions $k_0=1,1.5,2$.}
    \end{minipage}\hfill
    \begin{minipage}[b]{0.475\textwidth}
            \centering
    \includegraphics[width=0.95\textwidth]{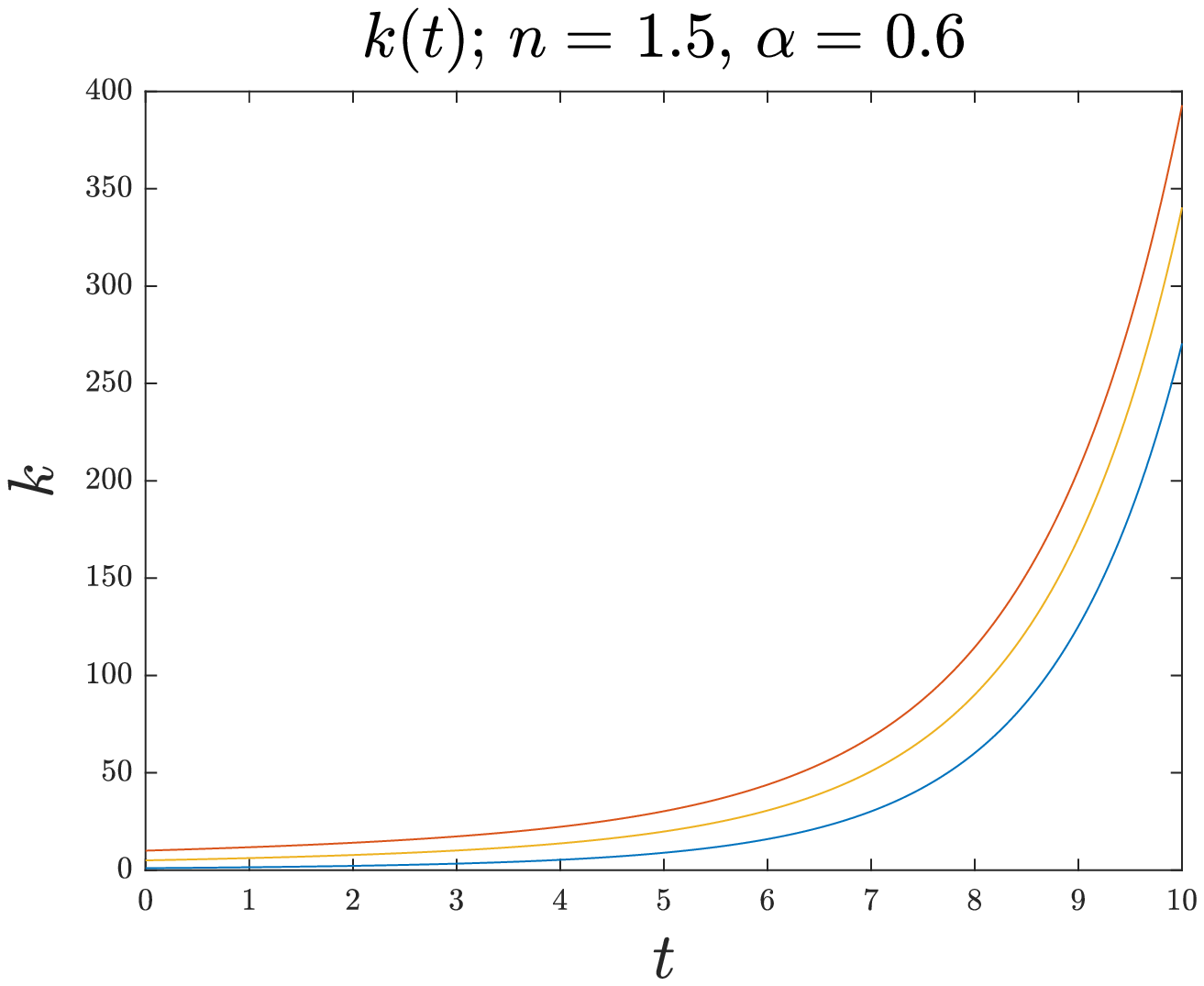}
  \caption*{(b) For $n>1$, we observe an exponential growth, independent on the initial conditions $k_0=1,5,10$.}
    \end{minipage}
    \caption{Numerical simulations of (\ref{eqn:thmCD}) for (a) $n<1$ (b) $n>1$. The value of $\alpha$ is displayed in the titles of each figure. The other values of the parameters are $\beta=n-\alpha$, $\gamma=0.7$, $s=0.4$, $L_0=1$.}
    \label{Figure_3}
\end{figure}

\begin{figure}[ht!]
    \centering
    \begin{minipage}[b]{0.475\textwidth}
            \centering
    \includegraphics[width=0.95\textwidth]{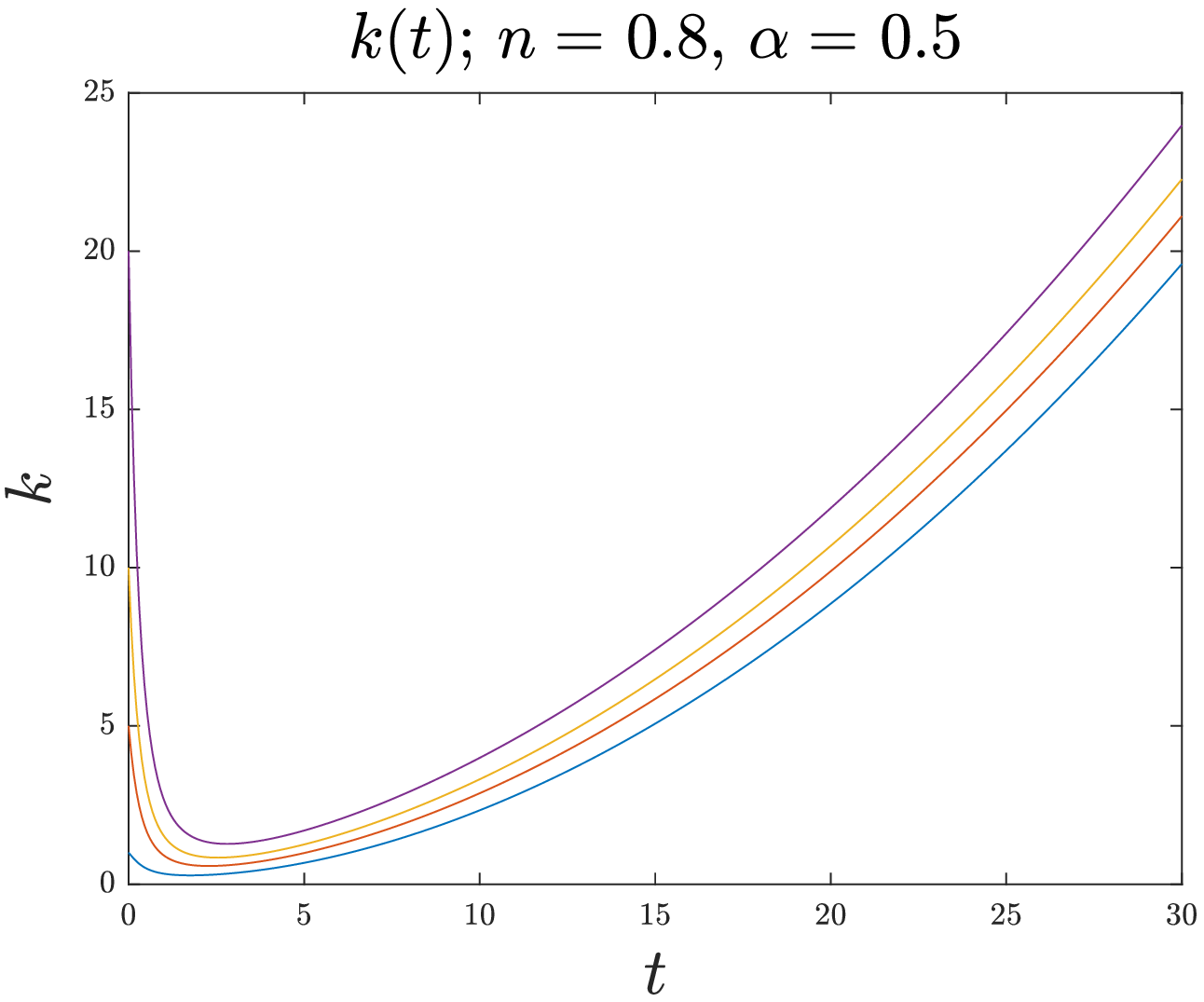}
  \caption*{(a) For $n<1$, we observe a rapid and short decrease, dependent on initial conditions, followed by a exponential growth starting for all initial conditions $k_0=1,5,10,20$.}
    \end{minipage}\hfill
    \begin{minipage}[b]{0.475\textwidth}
            \centering
    \includegraphics[width=0.95\textwidth]{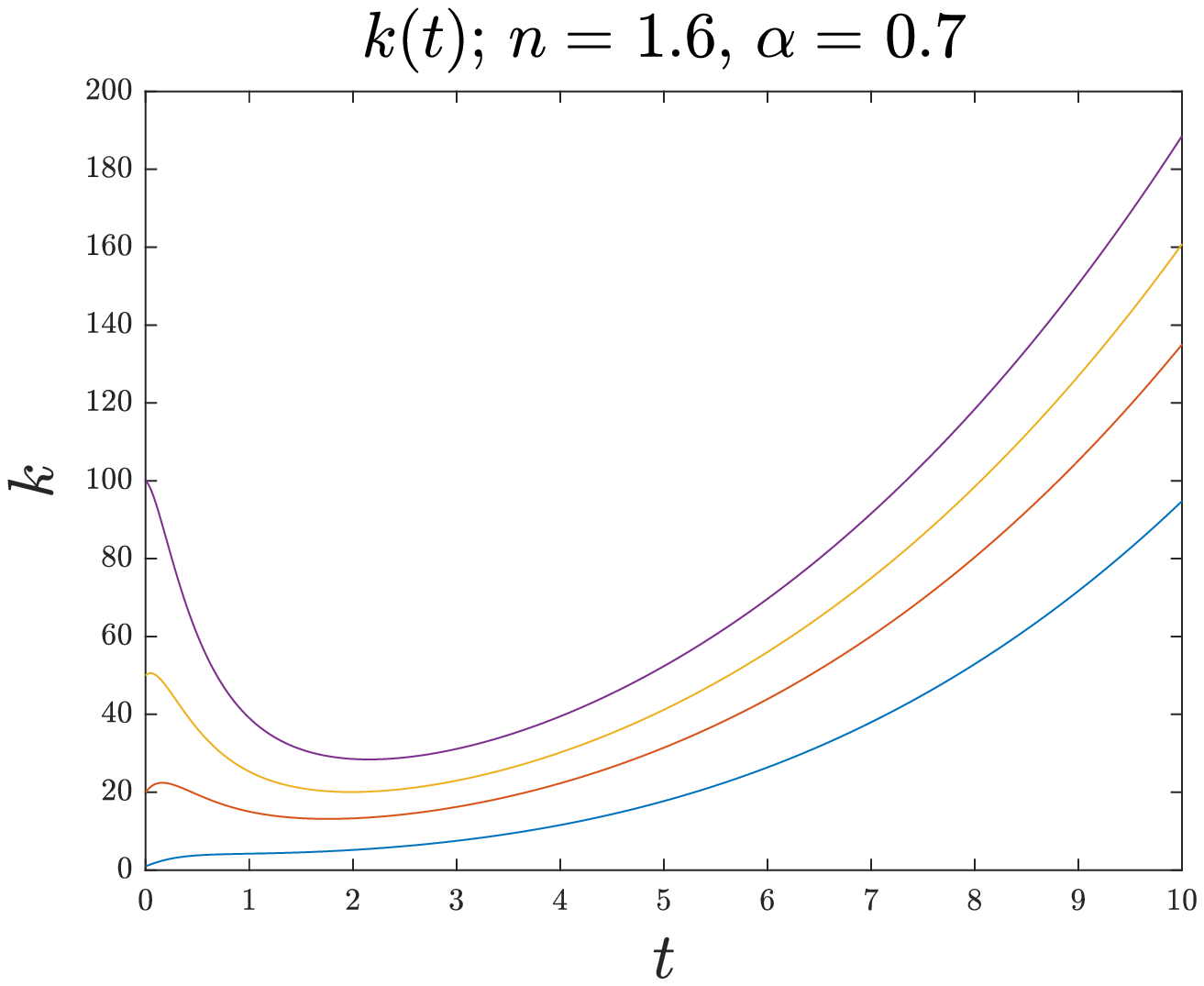}
  \caption*{(b)  For $n>1$, we observe a dependence on initial condition for the first part of the dynamics, followed by exponential growth, with initial conditions $k_0=1,20,50,100$.}
    \end{minipage}
    \caption{Numerical simulations of \eqref{eqn:thmCD2} for (a) $n<1$  (b) $n>1$.  The value of $\alpha$ is displayed in the titles of each figure. The other values of the parameters are $L_0=1$, $L_\infty=5$, $s=0.4$, $r=0.9$.}
    \label{Figure_4}
\end{figure}

\section{Conclusions}
The analysis of the Solow-Swan type models presents several stimulating mathematical challenges, which might be explored. In this work, we dwell on the case of non-constant returns to scale, providing an exact solution for the model that arise for the Cobb-Douglas production function. A more complicated case, namely the von Bertalanffy model, is also studied with similar results. Numerical simulations support the economical idea under the behaviour of the capital-labor ratio. Many issues remain open. One of this is, for instance, the study of the model for the CES production function (which does not satisfy the Inada conditions) with non-constant returns to scale or trying other, more exotic, production functions. We plan to explore these possibilities in the near future. 
\bigskip

\small
\noindent
\textbf{Acknowledgements:} NC and MS would like to thank the University of Pavia and the University of Trento, respectively, for supporting their research.
\bigskip%

\noindent
This research did not receive any specific grant from funding agencies in the public, commercial, or
not-for-profit sectors.

\bibliographystyle{plain} \footnotesize
\bibliography{references}
\end{document}